\documentclass[11pt]{article} %onecolumn (second format)

\usepackage{graphicx}
\usepackage{mathptmx}
\usepackage{amsmath}
\usepackage{amstext,amsbsy}
\usepackage{epsfig,multicol}
\usepackage{amssymb,latexsym}
\usepackage{amsthm}
\usepackage{multirow, rotating}
\usepackage[authoryear]{natbib}
\usepackage[pagebackref,hypertexnames=false, colorlinks, citecolor=blue, linkcolor=red]{hyperref}
\usepackage{hypernat}
\bibpunct{(}{)}{,}{a}{}{;}
\setcitestyle{aysep={}}

\newcommand{\ra}{\Rightarrow}
\newcommand{\ras}{\overset{\ast}{\Rightarrow}}

\theoremstyle{plain}
\newtheorem{theorem}{Theorem}[section]
\newtheorem{lemma}[theorem]{Lemma}
\newtheorem{corollary}[theorem]{Corollary}

\newtheorem{example}[theorem]{Example}
\newtheorem*{corollary*}{Corollary}{\bf}{\it}

\begin{document}

%\keywords{RNA secondary structure \and stochastic context-free grammar \and central limit theorem }

\title{Asymptotic distribution of motifs in a stochastic context-free grammar model of RNA folding}
\author{
Svetlana Poznanovi\'c, Christine E. Heitsch, 
 \vspace{.2cm} \\
School of Mathematics\\
Georgia Institute of Technology } 
\date{}
\maketitle
\begin{abstract}
We analyze the distribution of RNA secondary structures given by the Knudsen-Hein stochastic context-free grammar used in the prediction program Pfold. We prove that the distribution of base pairs, helices and various types of loops in RNA secondary structures in this probabilistic model is  asymptotically Gaussian, for a generic choice of the grammar probabilities. Our proofs are based on singularity analysis of probability generating functions. Finally, we use our results to discuss how this model reflects  the properties of some known ribosomal secondary structures.
\end{abstract}

%\noindent{\bf Keywords:}

%\noindent {\bf MSC Classification:}

{\renewcommand{\thefootnote}{}
\footnote{\emph{E-mail addresses}: svetlana@math.gatech.edu (Svetlana Poznanovi\'{c}), 
heitsch@math.gatech.edu (Christine E. Heitsch)} 
\footnotetext[1]{This work was supported by a BWF CASI grant to CEH. CEH was also supported in part by NIH NIGMS R01 GM08361.}

%[Distribution of motifs in a SCFG model of RNA folding]

%\subclass{92D20 \and 05A16 \and 60F05}

%{\renewcommand{\thefootnote}{}
%\footnote{\emph{E-mail addresses}: svetlana@math.gatech.edu (S. Poznanovi\'{c}), 
%heitsch@math.gatech.edu (C. E. Heitsch). \\

\section{Introduction} Knowing the base pairings  of an RNA sequence can reveal important information about the molecule's function but, unfortunately, experimental determination  of the secondary structure is too often nontrivial. For this reason, computational methods  have become a standard approach to RNA secondary structure prediction. Most of these prediction methods  are based on energy minimization~\citep{MT} and depend on the model for the folding free energy change.  In order to increase the prediction accuracy, the thermodynamic model has been refined over the years with the inclusion of hundreds of different parameters, most of them experimentally determined~\citep{NNDB}.  Alas, the prediction accuracy still varies widely~\citep{DCCG}. As an alternative, methods that use stochastic context-free grammars (SCFGs) have been developed~\citep{ED, SBHMSUH}. An advantage of these methods is that they can be augmented by phylogenetic and experimental information. One example of such a prediction program is Pfold~\citep{Pfold}.

When developing a prediction method based on a SCFG, several choices need to be made including the SCFG to be used and the set of probabilities for the grammar rules.
\citet{DE} performed an evaluation of the performance of several SCFGs in the prediction of secondary structures. The Knudsen-Hein grammar~\citep{SCFG} used in Pfold~\citep{Pfold} was found to be the most accurate one with prediction accuracy comparable to  energy minimization programs, while being significantly simpler than the other SCFGs tested. The authors of~\citep{DE} conclude that ``after exploring various alternative SCFG designs, we confirm that the Knudsen/Hein grammar is an excellent, simple framework in which to develop some probabilistic RNA analysis methods''. However,  the Sensitivity and the Positive Predicted Value of such predictions are still below 50\% for a lot of sequences, as given in~\citet[Table 3]{DE}.  In order to understand the potential for increasing prediction accuracy by changing the probabilities, in this work we analyze the probability distribution of RNA secondary structures generated by this grammar.

 There are two types of probability parameters: transmission probabilities, which are used for generation of the base pairs in the secondary structure, and emission probabilities which are used to generate the underlying sequence of nucleotides. The grammar defines a probability measure on the set of all pairs $(str, seq)$ of secondary structures $str$ and RNA sequences $seq$ of same length.  Then the basic way to predict a structure for a given RNA sequence $seq$ is to use the Cocke-Younger-Kasami (CYK) algorithm~\citep[see][]{Durbin} to compute the most probable pair $(str, seq)$.  
 
 The goal of this paper is to help clarify the effects of changing the probability parameters for the Knudsen-Hein SCFG. Specifically, we prove that the distributions of many biologically relevant motifs (helices, hairpins, multi-branch loops, etc.) are asymptotically Gaussian for almost all choices of the transmission probabilities for the grammar rules. In addition, we compute the asymptotic means and standard deviations as a function of these probabilities. A significant consequence of these results are relations between these distributions which are not affected by the change of the parameters (Corollary~\ref{col}). These relations are observed for the predicted structures of the ribosomal sequences (Section~\ref{discussion}) but do not hold for the native ribosomal structures. Consequently, the accuracy of the predictions for the long 16S and 23S sequences using the CYK algorithm cannot be significantly improved with a simple change of parameters. In particular we note that the strength of Pfold is in coupling the Knudsen-Hein grammar with phylogenetic information from sequence alignments.

The outline of the paper is as follows. In Section~\ref{prelim}, we give the definitions of secondary structure and the Knudsen-Hein SCFG and state our main results. In Section~\ref{singanal}, we illustrate the method of singularity analysis of generating functions on which our proofs are based.  In Section~\ref{main}, we derive the central limit theorems for various types of motifs and the asymptotic means as functions of the grammar probabilities. We additionally compute the expected number of multi-branch loops of a fixed degree and analyze the structure of the external loop. Finally, in Section~\ref{discussion}, we compare the theoretical results with the secondary structures from the Comparative RNA website~\citep{CRW} and the structures predicted for the same sequences using the CYK algorithm with the default Pfold parameters.

\section{Preliminaries} \label{prelim}

%\begin{definition}\label{secstr} 
A secondary structure of length $n$ is a graph with vertex set $\{1, 2, 3, \dots, n\}$, whose edge set consists of the edges $\{(k, k+1) : 1 \leq k \leq n-1\}$, together with a collection of edges $B$ called base pairs which satisfies the following conditions. For $(i, j), (k, l) \in B$,
\begin{enumerate}
\item  $j - i > \theta$ for some threshold $\theta > 0$,
\item $i \neq l$ and $i=k$ $\Leftrightarrow$ $j=l$,
\item $i< k < j$  $\Rightarrow$ $i < k < l < j$.
\end{enumerate}
%\end{definition}

The first condition reflects the fact that due to steric constraints, each hairpin in the secondary structure has to contain at least $\theta$ unpaired nucleotides. The second condition implies that each vertex (i.e. nucleotide) can belong to at most one base pair. Finally, the third condition excludes pseudoknots which are often considered to be a part of the tertiary structure of the RNA molecule and requires that two edges $(i,j)$ and $(k,l)$ in $B$  with $i<k$, either define separate domains (when $j<k$) or are nested (when $j < l$).  All secondary structures consist of the following basic motifs illustrated in Figure~\ref{loops}. A helix is a set of contiguous nested base pairs. A hairpin is a sequence of consecutive single-stranded nucleotides closed by a single base pair.   A bulge loop interrupts helices by having unpaired nucleotides in a single strand. It can be left or right, depending on the side on which the single stranded nucleotides appear. An internal loop separates two  helices by having unpaired nucleotides on both strands, while a multi-branch loop has three or more helices radiating from it. The single stranded nucleotides that are not enclosed by a base pair form an external loop.

\begin{figure}[hbt!]
\begin{center}
     \includegraphics[height=3cm]{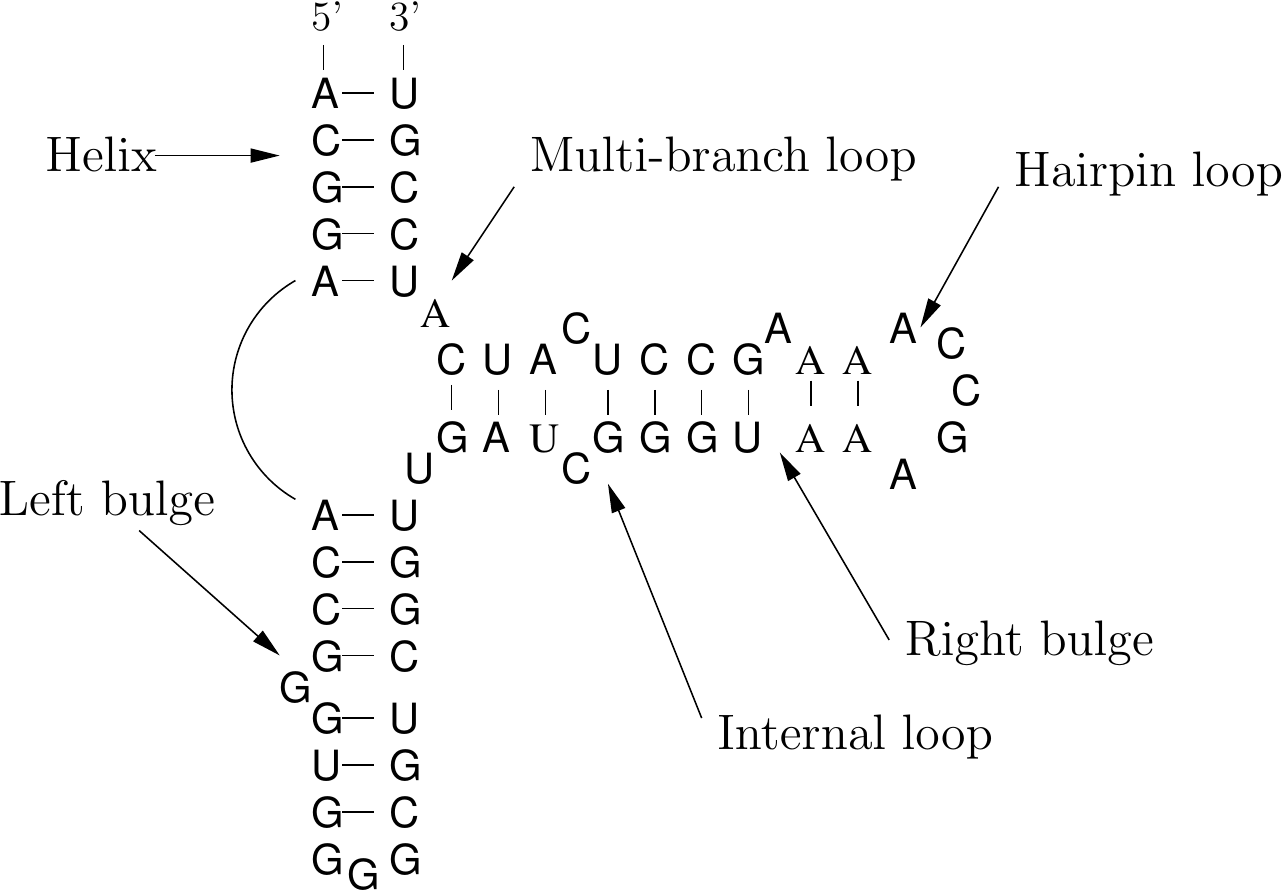}
  \caption{Helices and different types of loops in RNA secondary structures}
 \label{loops}
 \end{center}
\end{figure}

RNA secondary structures can be modeled using context-free grammars~\citep[see][]{Durbin}. The Knudsen-Hein grammar  which is used in the RNA secondary structure prediction program Pfold consists of nonterminal symbols $\{S, L, F\}$, terminal symbols $\{d, d', s\}$ and the rules
\begin{align*}
&S \rightarrow LS   \hspace{0.6cm} (p_1)   \hspace{0.5cm}\text{       or   } \hspace{0.5cm} L  \hspace{0.7cm}(q_1)  \\
&L \rightarrow dFd' \hspace{0.3cm} (p_2) \hspace{0.5cm}\text{       or   } \hspace{0.5cm} s  \hspace{0.8cm}(q_2)\\
&F \rightarrow dFd' \hspace{0.3cm}(p_3) \hspace{0.5cm}\text{       or   } \hspace{0.5cm} LS \hspace{0.5cm}(q_3).
\end{align*}
The numbers $p_i, q_i$, $i=1,2,3$ listed in parentheses are the probabilities for the production rules. They satisfy $p_i+q_i=1$ and $p_i, q_i >0$ and depend on the structures on which the grammar is trained. 

This grammar in non-ambiguous and each derivation corresponds to a unique secondary structure in which $j - i > 2$ for every base pair $(i,j)$. That is why in this paper by a secondary structure we will mean all graphs that satisfy the conditions in the definition of secondary structure for $\theta = 2$.  The terminal symbols $d$ and $d'$ correspond to left and right end nucleotides in a base pair, while $s$ corresponds to a single stranded nucleotide. Since secondary structures do not have pseudoknots, specifying the left and right ends of base pairs completely determines the whole structure. 
\begin{example} The simple hairpin given in Figure~\ref{derivation} is derived in the following way:
\[ S \overset{p_1}{\Rightarrow} LS \overset{q_2}{\Rightarrow} sS \overset{p_1}{\Rightarrow} sLS \overset{p_2}{\Rightarrow} sdFd'S \overset{p_3}{\Rightarrow} sddFd'd'S \overset{q_3}{\Rightarrow} sddLSd'd'S \overset{q_2}{\Rightarrow} \] \[sddsS d'd'S \overset{p_1}{\Rightarrow} sddsLSd'd'S \overset{q_2}{\Rightarrow}  sddssSd'd'S \overset{q_1}{\Rightarrow} sddssLd'd'S \overset{q_2}{\Rightarrow} \] \[sddsssd'd'S \overset{q_1}{\Rightarrow} sddsssd'd'L \overset{q_2}{\Rightarrow} sddsssd'd's.\] 
%Since the Knudsen-Hein grammar is non-ambiguous, the probability of this hairpin is $p_1^3p_2p_3q_1^2q_2^5q_3$.
\begin{figure}[hbt!]
\begin{center}
     \includegraphics[height=2cm]{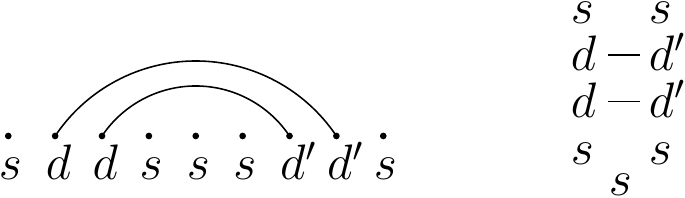}
\caption{A simple hairpin loop whose probability is $p_1^3p_2p_3q_1^2q_2^5q_3$}
 \label{derivation}
 \end{center}
\end{figure}
\end{example}

A stochastic grammar induces a probability distribution on the entire language if the sum of the probabilities of all the derivations is equal to 1. For each nonterminal symbol $N$, let $N(z)$ be the probability generating function of all secondary structures that can be generated starting from $N$, where $z$ records the number of nucleotides. In particular, if $n(M)$ is the number of nucleotides in a secondary structure $M$, we define
\[S(z) = \sum_{S \ras M} p(M) z ^ {n(M)}\] where $p(M)$ denotes the probability of the derivation of $M$ and the sum is over all secondary structures.  We can determine $S(z)$ by using a technique known as the Delest-Sch\"utzenberger-Viennot (DSV) method \citep{DSV}. For this we will need to work with $L(z)$ and $F(z)$ defined as
\[L(z) = \sum_{L \ras M} p(M) z ^ {n(M)}, \hspace{3cm} F(z) = \sum_{F \ras M} p(M) z ^ {n(M)},\]  where the sum is taken over all derivations $M$ that can be obtained starting from the nonterminals $L$ and $F$, respectively, and $p(M)$ denotes the probability of the derivation $M$. Through this technique the grammar can be converted into equations involving the generating functions $S(z)$, $L(z)$, $F(z)$. We get
\begin{align}
S(z) &= p_1 L(z)S(z) + q_1L(z) \notag\\
L(z) &= p_2 z^2 F(z) +q_2 z\\
F(z) &= p_3 z^2 F(z) + q_3 L(z) S(z). \notag
\end{align}
Eliminating $L(z)$ and $F(z)$, we get
\[ p_2 q_3 z^2 S(z)^2 -(1- p_1 q_2  z) (1 - p_3z^2)  S(z) + q_1 q_2 z (1-p_3 z^2) =0. \] Since $S(z)$ is a probabilistic generating function, it has a radius of convergence at least 1. Together with  $S(0) = 0$, this implies that \begin{equation} \label{S2} S(z)=\frac{(1- p_1 q_2  z) (1 - p_3z^2) - \sqrt{(1- p_1 q_2  z)^2 (1 - p_3z^2)^2 - 4 p_2 q_1 q_2 q_3 z^3 (1-p_3 z^2)}} {2p_2 q_3 z^2} \end{equation}
To determine when this grammar generates a probabilistic language, we find when $S(1)=1$. The condition
\begin{equation} \frac{(1- p_1 q_2 ) (1 - p_3) - \sqrt{(1- p_1 q_2)^2 (1 - p_3)^2 - 4 p_2 q_1 q_2 q_3(1-p_3)}} {2p_2 q_3} =1 \end{equation} is equivalent to
\begin{equation}|p_2-q_1q_2| = q_1q_2-p_2. \end{equation} Recalling that $p_2=1-q_2$, this reduces to
\begin{equation} \label{properties}(1+ q_1) q_2 \geq 1.\end{equation}

Our main result is a central limit theorem for the number of helices and the various types of loops generated by the Knudsen-Hein SCFG. 
\begin{theorem} \label{mainth} Let $\mathbb{X}_n$ be the number of base pairs, or helices,  or loops of a fixed type in a random secondary structure with $n$ nucleotides. If the probabilities are such that $f(p_1, p_2, p_3) \neq 0$ for a certain function $f$, then there exist nonzero constants $\mu$ and $\sigma$ such that the normalized random variables $$\mathbb{X}^*_n = \frac{\mathbb{X}_n-\mu n}{\sqrt{n \sigma^2}}$$ converge in distribution to a Gaussian variable with a speed of convergence $O(\frac{1}{\sqrt{n}})$. That is, we have 
\[\lim_{n \rightarrow \infty}  \mathbb{P} \left( \mathbb{X}^*_n< x\right) = \frac{1}{\sqrt{2\pi}} \int_{-\infty}^{x} e^{-\frac{c^2}{2}} \; dc \]
and
\[\sup_{x \in \mathbb{R}} \left| \mathbb{P} \left( \mathbb{X}^*_n< x\right) - \frac{1}{\sqrt{2\pi}} \int_{-\infty}^{x} e^{-\frac{c^2}{2}} \; dc \right| \leq O \left(\frac{1}{\sqrt{n}}\right).\]
 \end{theorem} 
 
The function $f$ which appears in the conditions of Theorem~\ref{mainth} is discussed in Section~\ref{funf}, where we explain why $f(p_1, p_2, p_3) \neq 0$ for all probabilities except for a set of measure zero, so that the result holds for almost all choices of probabilities. Theorem~\ref{mainth} is proved in Section~\ref{main}, where the different types of motifs are considered separately. The proof is based on singularity analysis of bivariate generating functions. In the following section, we illustrate this method by obtaining the asymptotic estimate for the coefficients of $S(z)$. 

The constants $\mu$ from Theorem~\ref{mainth} are given as functions of the probabilities in Section~\ref{main} for all motifs. A surprising fact is that the following relations between them hold independently of the probability parameters.

\begin{corollary*} If $p_i, q_i >0$
\begin{itemize}
\item [(i)] $\mathbb{E}(\mathbb{X}_n^{lb}) = \mathbb{E}(\mathbb{X}_n^{rb})$,
\item [(ii)] $\mathbb{E}(\mathbb{X}_n^m) = \frac{1}{4}\mathbb{E}(\mathbb{X}_n^{hel})(1+o(n))$,
\item [(iii)] $\mathbb{E}(\mathbb{X}_n^{hp})=(\mathbb{E}(\mathbb{X}_n^{i}) + \mathbb{E}(\mathbb{X}_n^{m})) (1+ o(n))$,
\item [(iv)] $\mathbb{E}(\mathbb{X}_n^{m})=(\mathbb{E}(\mathbb{X}_n^{lb}) + \mathbb{E}(\mathbb{X}_n^{i})) (1+ o(n))$,
\item [(v)] $\mathbb{E}(\mathbb{X}_n^{m,r+1}) < \frac{1}{2}\mathbb{E}(\mathbb{X}_n^{m,r})(1+o(n))$, $r \geq 2$.
\end{itemize}
where the superscripts $lb, rb, m, hel, hp$, and $i$ denote left bulges, right bulges, multi-branch loops, helices, hairpins, and internal loops respectively, while $\mathbb{X}_n^{m, r}$ is the number of multi-branch loops of degree $r$ in a random secondary structure with $n$ nucleotides. 
\end{corollary*}
We find the invariance of these relations under parameter change especially interesting because it illustrates that the potential for changing the distributions by changing of different motifs is limited. This is important when the problem at hand is to model  structures which do not satisfy the same equalities.

\section{Singularity analysis} \label{singanal}

The total probability of all structures with $n$ nucleotides is given by $[z^n]S(z)$, the coefficient of $z^n$ in $S(z)$. This result will be needed later, so we derive it here as our basic example of asymptotic analysis related to this grammar. We use the following theorem of~\citet{FO} to determine the asymptotic growth of the coefficients of $S(z)$.
\begin{theorem}[\citet{FO}] \label{fo}
Assume that $S(z)$ has a singularity at $z = \rho >0$, is analytic in the region $\Delta \setminus \{\rho\}$, depicted in Figure~\ref{deltaregion}, and that as  $z \rightarrow \rho$
in $\Delta$,
$S(z) \sim K(1 - z/\rho)^c$, for some constants $K \neq 0$ and $c \neq 0, 1, 2, \dots$.

Then, as $n \rightarrow \infty$, 
\[[z^n]S(z) \sim \frac{K}{\Gamma(-c)} n^{-c-1} \rho^{-n},\]
where $\Gamma(z)$ denotes the classical gamma function.
\end{theorem}

\begin{figure}[hbt!]
\centering
     \includegraphics[height=4cm]{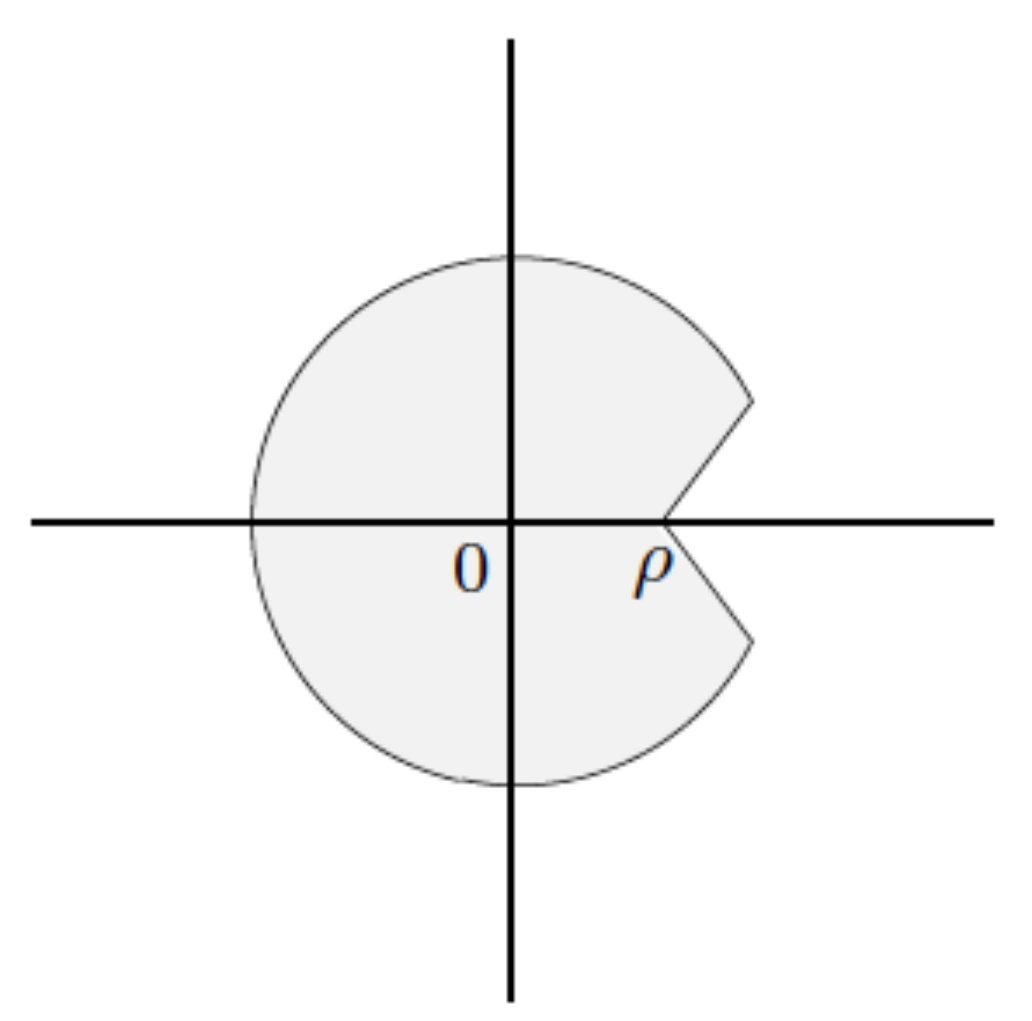}
  \caption{The function $S(z)$ needs to be analytic at all points in a region $\Delta$ of the depicted shape except at $\rho$}
 \label{deltaregion}
\end{figure}

Set \begin{equation} \label{R} R(z) = (1- p_1 q_2  z)^2 (1 - p_3z^2) - 4 p_2 q_1 q_2 q_3 z^3.\end{equation}

From the explicit formula for $S(z)$ given in~\eqref{S2}, we see that the singularities of $S(z)$ are the zeros of the polynomial $P(z) = (1-p_3z^2)R(z)$, and in fact the dominant singularity is a root of $R(z)$, which follows from the following two lemmas. 

\begin{lemma}\label{lm1}  If $p_i, q_i >0$, one of the roots of $R(z)$ of smallest modulus  is a positive real number. \end{lemma}
\begin{proof}
Since $S(z)$ is a probability generating function, it has a radius of convergence at least 1. The fact that the coefficients of $S(z)$ are positive implies that it has a positive real singularity equal to its radius of convergence. From~\eqref{S2}, we see that this singularity must be a root of the polynomial 
\[(1- p_1 q_2  z)^2 (1 - p_3z^2)^2 - 4 p_2 q_1 q_2 q_3 z^3 (1-p_3 z^2) = R(z) (1-p_3 z^2). \] Since $R(0) =1 >0$ and $R(1/\sqrt{p_3}) <0 $, $R(z)$ has a real zero in the interval $(0,1/\sqrt{p_3})$. Therefore the smallest positive real singularity of $S(z)$ must come from the zeros of $R(z)$, which implies that among the zeros of smallest modulus of $R(z)$, one is positive and real. 
\end{proof}

From now on, let $\rho_0$ be the root with smallest modulus of $R(z)$ which is a positive real number. The following properties of $\rho_0$ will be used in the proofs that follow.

\begin{lemma}\label{lm} $\rho_0$ is the unique root of $R(z)$ on the circle $\{z : |z| =\rho_0 \}$. Moreover, $$1< \rho_0 < \min{\{ 1/ p_1q_2, 1/\sqrt{p_3} \}} \hspace{0.5cm} \text{ and } \hspace{0.5cm} R'(\rho_0)<0.$$
\end{lemma}

\begin{proof}
In the proof of Lemma~\ref{lm1}, we have already shown that $\rho_0 < 1/\sqrt{p_3}$. The fact that $\rho_0<1/p_1 q_2$ also follows from  $R(0)>0$ and $R(1/p_1q_2)<0$. Suppose that $R(z)$ has two complex roots $w, \bar{w}$, with $|w|=\rho_0$. Then, by the triangle inequality,
\[ \left|1- p_1 q_2  w\right| > 1- p_1 q_2  \rho_0 >0 \hspace{1cm} \text{and} \hspace{1cm} \left|1 - p_3w^2 \right| > 1 - p_3\rho_0^2 >0,\]  where the  inequalities are strict because $w$ is not real. This contradicts
\begin{align*}\left|1- p_1 q_2  w\right|^2 \left|1 - p_3w^2 \right| &=4 p_2 q_1 q_2 q_3 |w|^3 \\ &= 4 p_2 q_1 q_2 q_3 \rho_0^3=(1- p_1 q_2  \rho_0)^2 (1 - p_3\rho_0^2).\end{align*}  
Similarly, we get a contradiction if we assume that $R(-\rho_0) = 0$.
Lastly, we compute $$R'(z) = -2p_1q_2(1- p_1 q_2  z) (1 - p_3z^2) - 2p_3 z (1- p_1 q_2  z)^2 - 12p_2 q_1 q_2 q_3 z^2,$$from where it is clear that $R'(\rho_0) <0$.
\end{proof}

As a consequence, if we set
\[ Q(z) = \frac{(1- p_1 q_2  z) (1 - p_3z^2)} {2p_2 q_3 z^2}\] and
\[P(z) = P_1(z)(1-\frac{z}{\rho_0}),\]
then \[S(z) - Q(\rho_0) = \frac{-\sqrt{P_1(\rho_0)}}{2p_2 q_3 \rho_0^2} (1-\frac{z}{\rho_0})^{1/2} + O(1-\frac{z}{\rho_0})\] when $z \rightarrow \rho_0$.

The coefficients in the expansion of $S(z)$ are the same as in the expansion of $S(z)-Q(\rho)$, except for the first one. From Theorem~\ref{fo}, we get 

\begin{equation} \label{sn} [z^n]S(z) \sim -\frac{\sqrt{(1- p_1 q_2  \rho_0) (1 - p_3\rho_0^2)(3-p_1q_2 \rho_0 -p_3\rho_0^2-p_1p_3q_2\rho_0^3)}}{2p_2q_3\rho_0^2 \Gamma(-1/2)} n^{-3/2} \rho_0^{-n}. \end{equation}

\section{Asymptotic distributions of substructures} \label{main}

In this section we prove central limit theorems for various RNA secondary structure motifs for generic choices of the grammar probabilities. We will use the following theorem~\citep[Theorem IX.12]{FS} which we state specialized for our purposes.
\begin{theorem}[\citet{FS}]\label{fl}
Let $G(z,u)$ be a function that is bivariate analytic at $(z,u)=(0,0)$ and has non-negative coefficients and let  $\mathbb{X}_n$ be a random variable such that
\[\mathbb{P}(\mathbb{X}_n = k) = \frac{[z^n u^k] G(z,u)}{[z^n] G(z,1)}. \]
If the technical conditions $(i) - (iii)$ listed below are satisfied, then there exist constants $\mu$ and $\sigma$  such that the normalized random variable 
\[\mathbb{X}^*_n=  \frac{\mathbb{X}_n-\mu n}{\sqrt{n\sigma^2}}\] converges in distribution to a Gaussian variable with a speed of convergence $O(\frac{1}{\sqrt{n}})$. \\ 
\noindent The technical conditions are
\begin{itemize}
\item[(i)] There exist  functions $A, B, C$ analytic in a domain $\mathcal{D} = \{ |z| < r\} \times \{|u-1| < \epsilon\}$ such that
\[ G(z,u) = A(z,u) + B(z,u)C(z,u)^{1/2} \]
for all $(z,u) \in \{ |z| < r_0\} \times \{|u-1| < \epsilon\}$ for some $r_0 <r$. Furthermore, assume that in $|z| < r$, there exists a unique root $z_0$ of the equation $C(z,1)=0$ and that $B(z_0,1) \neq 0$, 

\item[(ii)] $C_{1,0} C_{0,1} \bigg|_{z=z_0, u=1} \neq 0$, where $C_{i,j} = \frac{\partial^{i+j}}{\partial z^i \partial u^j}C,$

\item[(iii)] \begin{equation} \label{variability} z_0 C_{1,0}^2 C_{0,2} - 2 z_0 C_{1,0} C_{1,1} C_{0,1} + z_0 C_{2,0}C_{0,1}^2  + C_{0,1}^2 C_{1,0} + z_0 C_{0,1} C_{1,0}^2 \bigg|_{z=z_0, u=1} \neq 0. \end{equation} 
\end{itemize}
\noindent The constants $\mu$ and $\sigma$ are given by:
\begin{eqnarray}
\mu &=& \frac{C_{0,1}}{z_0 C_{1,0}} \bigg|_{z=z_0, u=1}\\
\sigma^2 &=&  \frac{z_0 C_{1,0}^2 C_{0,2} - 2z_0 C_{1,0} C_{1,1} C_{0,1} + z_0 C_{2,0}C_{0,1}^2  + C_{0,1}^2 C_{1,0} + z_0 C_{0,1} C_{1,0}^2 }{z_0^2 C_{1,0}^3 } \bigg|_{z=z_0, u=1} \label{stdev}
\end{eqnarray}
\end{theorem}

\subsection{Base pairs} 
To find the distribution of base pairs, we first find the bivariate generating function $S(z,u)$ where $u$ marks the base pairs. A base pair is added precisely when the rules $L \rightarrow dFd'$ and $F \rightarrow dFd'$ are used. So, $S(z,u)$ is the solution of the system
\begin{align}
S(z,u) &= p_1 L(z,u)S(z,u) + q_1L(z,u) \notag\\
L(z,u) &= p_2 z^2 u F(z,u) +q_2 z \\
F(z,u) &= p_3 z^2 u F(z,u) + q_3 L(z,u) S(z,u). \notag
\end{align}
Similarly as before, we can find an explicit formula for $S(z,u)$:
\[ S(z,u) = Q(z,u) - \frac{\sqrt{C^{bp}(z,u)}}{2 p_2q_3 z^2 u}\] where
\begin{eqnarray}
Q(z,u)&=&\frac{(1-p_1q_2z)(1-p_3z^2u)}{2 p_2q_3 z^2 u}, \notag\\
C^{bp}(z,u)&=&(1-p_1q_2z)^2(1-p_3z^2u)^2 - 4 p_2 q_1 q_2 q_3  z^3 u (1-p_3 z^2 u). \label{Cbp}
\end{eqnarray}
\begin{theorem} \label{bps}
Let $\mathbb{X}_n^{bp}$ be a random variable counting the number of basepairs in a secondary structure with $n$ nucleotides. If the probabilities $p_i, q_i>0$, $1 \leq i \leq 3$ are such that the polynomial $C^{bp}(z,u)$ given in~\eqref{Cbp} satisfies the condition~\eqref{variability}, then $\mathbb{X}_n^{bp}$ after standardization converges to a Gaussian variable.  The mean and standard deviation of $\mathbb{X}_n^{bp}$ are asymptotically linear in $n$. In particular,
\[ \mathbb{E}(\mathbb{X}_n^{bp}) \sim \frac{ \alpha}{\gamma} n,\] where
\begin{equation} \label{alphagamma} \alpha =  1-p_1q_2\rho_0, \hspace{1cm} \gamma = 3-p_1q_2\rho_0-p_3\rho_0^2-p_1p_3q_2\rho_0^3. \end{equation}
The first order approximation of the standard deviation is given by~\eqref{stdev} for $C=C^{bp}$ and $z_0=\rho_0$.
\end{theorem}
\begin{proof}

The random variables associated to $G(z,u)= z^2S(z,u)$ are the same as the ones associated to $S(z,u)$, only shifted in index. So, we will work with the function $G(z,u)$  and we will prove that it satisfies the conditions in Theorem~\ref{fl}. The functions $A(z,u) = z^2 P(z,u)$, $B(z,u) = -\frac{1}{2p_2q_3 u}$, and $C(z,u)=C^{bp}(z,u)$ are clearly analytic in the domain $\mathbb{C} \times \{|u-1| < \epsilon\}$ for  small $\epsilon>0$. Using  Lemma~\ref{lm}, we get

\begin{eqnarray*}
C_{1,0}(\rho_0,1)&=& R'(\rho_0)(1-p_3 \rho_0^2) <0 \\ 
C_{0,1} (\rho_0,1)%&=& -2p_3\rho_0^2(1-p_3\rho_0^2)(1-p_1q_1\rho_0)
%^2-4p_2q_1q_2q_3\rho_0^3(1-p_3\rho_0^2)+4p_2p_3q_1q_2q_3\rho_0^5\\
&=& -p_3\rho_0^2(1-p_3\rho_0^2)(1-p_1q_2\rho_0)^2 - 4p_2q_1q_2q_3\rho_0^3(1-p_3\rho_0^2) \\
&=& -4p_2q_1q_2q_3\rho_0^3<0.
\end{eqnarray*}
So, condition $(ii)$ is satisfied. By the analytic implicit function theorem, there exists an analytic function $\rho(u)$ defined on some neighborhood of $\rho_0$ such that $R(z,u) = 0$ for $(z,u)$ in a small polydisc $\Delta(\rho_0,  1, \epsilon)$ if and only if $z=\rho(u)$. Since, by Lemma~\ref{lm}, $\rho(1) = \rho_0 >1$,  $\epsilon$ can be chosen so that $\rho(u) >1$. 

We claim that if $|u-1|<\epsilon$, $z=\rho(u)$ is the root of smallest  modulus of $R(z,u)$ as a polynomial in $z$. There is a neighborhood of $u=1$ such that $z=\rho(u)$ is the unique zero of smallest modulus of $R(z,u)$. Otherwise, there exists a sequence $u_n \rightarrow 1$ and $\xi_n  \neq \rho(u_n)$ with $|\xi_n| \leq |\rho(u_n)|$ and $R(\xi_n, u_n) =0$. By passing to a subsequence, which we still denote by $(\xi_n)$, we obtain that there exists some $\xi_0$ such that $\lim_{n \rightarrow \infty} \xi_n = \xi_0$. By continuity, $R(\xi_0,1)=0$ and $|\xi_0| \leq \rho(1)$. Hence, by uniqueness, $\xi_0 = \rho(1)$. This contradicts the uniqueness of the solution $z=\rho(u)$ of $R(z,u)=0$ in a neighborhood of $u=1$ guaranteed by the implicit function theorem.

Finally, choose $\epsilon$ to be small enough so that $R(z)$ has a unique zero in $ |z| < \rho_0 + \epsilon$. Setting $r=\rho_0+\epsilon$ and $r_0 =1$, such that $R(z)$ has a unique zero make condition $(i)$ satisfied. Indeed, for $\mathcal{D} = \{ |z| < r\} \times \{|u-1| < \epsilon\}$, 
$C(z,1) = R(z)(1-p_3z^2)$ has a unique zero in $\mathcal{D}$ and clearly $B(\rho_0,1) \neq 0$. 

\end{proof}

\subsection{Helices}

A helix is started when the rule $L \rightarrow dFd'$ is used. If $u$ marks the number of helices in the secondary structure, the relation between the probability generating functions is
\begin{align*}
S(z,u) &= p_1 L(z,u)S(z,u) + q_1L(z,u), \\
L(z,u) &= p_2 z^2 u F(z,u) +q_2 z ,\\
F(z,u) &= p_3 z^2 F(z,u) + q_3 L(z,u) S(z,u),
\end{align*}
and therefore, 
\[S(z,u) = Q(z,u) - \frac{\sqrt{C^{h}(z,u)}}{2 p_2q_3 z^2 u},\]
where 
\begin{eqnarray}
Q(z,u)&=&\frac{(1-p_1q_2z)(1-p_3z^2)}{2 p_2q_3 z^2 u}, \notag\\
C^{h}(z,u)&=&(1-p_1q_2z)^2(1-p_3z^2)^2 - 4 p_2 q_1 q_2 q_3  z^3 u (1-p_3 z^2). \label{Ch}
\end{eqnarray}
\begin{theorem} \label{hel}
Let $\mathbb{X}_n^{hel}$ be the number of helices in a random secondary structure with $n$ nucleotides. If the probabilities $p_i, q_i>0$, $1 \leq i \leq 3$ are such that the polynomial $C^{h}(z,u)$ given in~\eqref{Ch} satisfies the condition~\eqref{variability}, then $\mathbb{X}_n^{h}$ after standardization converges to a Gaussian variable.  In particular,
\[ \mathbb{E}(\mathbb{X}_n^{hel}) \sim  \frac{\alpha \beta}{\gamma} n,\] where $\alpha$ and $\gamma$ are given by~\eqref{alphagamma} and
\begin{equation} \label{beta} \beta = 1-p_3\rho_0^2. \end{equation}
The first order approximation of the standard deviation is given by~\eqref{stdev} for $C=C^{hel}$ and $z_0=\rho_0$.
\end{theorem}
\begin{proof} Similarly as in the proof of Theorem~\ref{bps}, the conditions $(i)$ and $(ii)$ from Theorem~\ref{fl} are satisfied for the function $G(z,u) = z^2 S(z,u)$. Condition $(iii)$ is satisfied by assumption. 
\end{proof}

\subsection{Loops} In this subsection, let $S(z,x,y,u,v,w)$ be the multivariable probability generating function for RNA structures where $x$ marks hairpin loops, $y$ marks multi-branch loops, $u$ marks left bulges, $v$ marks right bulges, and $w$ marks internal loops.

A loop starts exactly when a helix ends, so each application of the rule  $F \rightarrow LS$ starts one loop. 
The loop started will be a hairpin loop if this rule is followed by $LS \ras s^n$, $ n\geq 2$. To find the probabilities of a  hairpin loop of length $n \geq 2$ we note that
\begin{eqnarray*}
P(LS \ras s^n) &=& P(L \ra s) P(S \ras s^{n-1}) \\ &=& q_2 P(S \ra LS) P(L \ra s) P(S \ras s^{n-2} )\\ &=& p_1 q_2^2 P(S \ras s^{n-2}) \\ &=& q_1 q_2^2 (p_1 q_2)^{n-2}
\end{eqnarray*}
Therefore the probability generating function for the hairpin loops that could be formed is
\[ \sum_{n=2}^{\infty} q_1 q_2^2 (p_1 q_2)^{n-2} z^n = \frac{q_1 q_2^2 z^2}{1 - p_1 q_2 z} =:H_h. \]
Right bulges are formed when the derivation that follows is of the form $LS \ras dFd' s^{l}$, $ l \geq 1$. Their probability is 
\[P(LS \ras dFd' s^{l} ) = P(L \rightarrow dFd' ) P( S \ras s^l) = p_2 q_1 q_2 (p_1 q_2)^{l-1} \] and their contribution to the generating function is 
\[\sum_{l=1}^{\infty}  z^{l+2} p_2 q_1 q_2 (p_1 q_2)^{l-1} F = \frac{p_2 q_1 q_2   z^3}{ 1- p_1 q_2 z} F =: H_b F .\] Similarly, left bulges are formed by applications of rules that yield $LS \ras s^ k dFd' $, $ k \geq 1$. The probability of the left bulges together with all successive derivations is
\begin{eqnarray*}
P(LS \ras s^ k dFd' ) &=& P(LS \ras s^k S) P( S \rightarrow L)P(L \rightarrow dFd') \\ &=&  p_2 q_1 P(LS \ras s^k S) \\
&=& p_2 q_1 P(LS  \ras sS) P(S \rightarrow LS) P(LS \ras s^{k-1}S) \\
 &=&p_2 q_1  q_2 (p_1 q_2)^{k-1}.
 \end{eqnarray*} The part of the generating function that corresponds to the left bulges is 
\[\sum_{k=2}^{\infty}  z^{k+2} p_2 q_1  q_2 (p_1 q_2)^{k-1} F  = \frac{p_2 q_1  q_2  z^3}{1- p_1 q_2  z} F = H_b F .\]
Internal loops are created when the rule $F \rightarrow LS$ is followed by $LS \ras s^k dFd' s^{l}$, for some $ k, l \geq 1$.
\begin{eqnarray*}
P(LS \ras s^k dFd' s^{l}) &=& P(LS \ras s^k S) P( S \rightarrow LS) P( LS \ras dFd' s^{l}) \\
&=&   q_2 (p_1 q_2)^{k-1} p_1 p_2 q_1 q_2 (p_1 q_2)^{l-1} \\
&=& p_1 p_2 q_1 q_2^2 (p_1 q_2)^{k-1}  (p_1 q_2)^{l-1} \end{eqnarray*}
and their contribution to the generating function is
\[ \sum_{k=1}^{\infty} \sum_{l=1}^{\infty} z^{k+l+2} p_1 p_2 q_1 q_2^2 (p_1 q_2)^{k-1}  (p_1 q_2)^{l-1} F =  \frac{p_1 p_2 q_1 q_2^2  z^4}{(1-p_1q_2  z)^2} F =: H_i F. \]
The remaining part of $LS$ corresponds to the substructures that begin with a multi-branch loop. Their contribution is
\[H_m:=LS - H_h - 2H_bF - H_iF. \]
Using this, the translation of the grammar rules yields the system
\begin{eqnarray*}
S &=& p_1 LS + q_1 L,\\
L &=& p_2 z^2 F + q_2 z,\\
F &=& p_3 z^2 F + q_3 (xH_h + uH_bF+vH_bF + wH_iF+ yH_m),
\end{eqnarray*}
and, by eliminating $F$ and $L$, we get that $S(z,x,y,u,v,w)$ is a solution to the quadratic equation
\begin{eqnarray*} p_2 q_3 z^2
y S^2 &+& ( p_1p_2q_3  z^2 xH_h -  p_1 p_2 q_3z^2 yH_h + p_1 q_2 z H_e-H_e ) S \\ &+& (p_2 q_1q_3 z^2 x H_h -  p_2q_1q_3 z^2 y H_h + q_1 q_2 z H_e) = 0\end{eqnarray*}
where
\[H_e := 1-p_3z^2+q_3(y-u)H_{b} + q_3(y-v)H_{b}+q_3(y-w)H_i.\]

\begin{theorem} \label{lps}
Let $\mathbb{X}^{hp}_n$, $\mathbb{X}^{lb}_n$, $\mathbb{X}^{rb}_n$, $\mathbb{X}^i_n$, and $\mathbb{X}^m_n$ be  the number of hairpin loops, left bulges, right bulges, internal loops,  and multi-branch loops in a random secondary structure with $n$ nucleotides, respectively.  For $\star \in \{hp, lb, rb, i, m\}$, if the probabilities $p_i, q_i>0$, $1 \leq i \leq 3$ are such that a certain polynomial $C^{\star}(z,u)$  satisfies the condition~\eqref{variability}, then  $\mathbb{X}^{\star}_n$ after standardization converges to a Gaussian variable. The approximate expectations are explicitly given by
\begin{eqnarray*}
\mathbb{E}(\mathbb{X}_n^{hp}) &\sim& \frac{ (1+p_1q_2\rho_0) \alpha \beta }{4 \gamma} n, \\
\mathbb{E}(\mathbb{X}_n^{lb}) = \mathbb{E}(\mathbb{X}_n^{rb}) &\sim& \frac{\alpha^2 \beta}{4\gamma} n, \\
\mathbb{E}(\mathbb{X}_n^{i}) &\sim& \frac{p_1q_2\rho_0 \alpha \beta}{4 \gamma} n, \\
\mathbb{E}(\mathbb{X}_n^{m}) &\sim& \frac{ \alpha \beta}{4\gamma} n, 
\end{eqnarray*}
where $\alpha, \beta$, and $\gamma$ are given by~\eqref{alphagamma} and~\eqref{beta}. The first order approximations of the standard deviations are given by~\eqref{stdev} for $C=C^{\star}(z,u)$ and $z_0=\rho_0$.
\end{theorem}

\begin{proof}
By setting $y=u=v=w=1$, for hairpins we get that 
\[S(z,x) = Q(z,x) - \frac{\sqrt{C^{hp}(z,x)}}{2p_2q_3z^2(1-p_1q_2z)}, \] where
\begin{eqnarray*}
Q(z,x)&&=\frac{(1-p_1q_2z)^2(1-p_3z^2) - p_1p_2q_1q_2^2q_3z^4(x-1)}{2p_2q_3z^2(1-p_1q_2z)},\\
C^{hp}(z,x)&&=\left( p_1p_2q_1q_2^2q_3z^4(x-1)- (1-p_1q_2z)^2(1-p_3z^2) \right)^2 \\&& - 4 p_2 q_1 q_2 q_3  z^3  (1-p_1q_2 z) \left( p_2q_1 q_2 q_3 z^3 (x-1) + (1-p_1q_2z) (1-p_3z^2) \right).
\end{eqnarray*}
To prove the claim for $\mathbb{X}_n^{hp}$, we work with the function $G(z,x) = z^2 S(z,x)$. The functions in condition $(i)$ of Theorem~\ref{fl}, are $A(z,x) = P(z,x)z^2$, $B(z,x) = -\frac{1}{2p_2 q_3 (1-p_1q_2z)}$, and $C^{hp}(z,x)$. They are all analytic in some polydisc around $(0,1)$.  Since \begin{eqnarray*}C^{hp}(z,1) &=& (1-p_1q_2z)^4(1-p_3z^2)^2 - 4 p_2 q_1 q_2 q_3  z^3  (1-p_1q_2 z)^2 (1-p_3z^2) \\ &=& R(z) (1-p_1q_2 z)^2 (1-p_3z^2), \end{eqnarray*}  it follows from Lemma~\ref{lm} that the smallest zero of $C^{hp}(z,1)$ is $\rho_0$ and that 
\begin{eqnarray*}C^{hp}_{1,0}(\rho_0,1)&=&R'(\rho_0) (1-p_1q_2 \rho_0)^2 (1-p_3 \rho_0^2),  \\
C^{hp}_{0,1} (\rho_0,1)&=&-2p_1p_2q_1q_2^2q_3\rho_0^4(1-p_1q_2 \rho_0)^2 (1-p_3 \rho_0^2) - 4 p_2^2q_1^2q_2^2q_3^2 \rho_0^6 (1-p_1q_2 \rho_0)
\end{eqnarray*} are both negative.

By setting $x=y=v=w=1$, for left bulges we get that 
\[S(z,u) = Q(z,u) - \frac{\sqrt{C^{lb}(z,u)}}{2p_2q_3z^2(1-p_1q_2z)}, \] where
\begin{eqnarray*}
Q(z,u)&=&\frac{(1-p_1q_2z)(1-p_3z^2) + p_2q_1q_2q_3z^3(1-u)}{2p_2q_3z^2},\\
C^{lb}(z,u)&=&\left( (1-p_1q_2z)^2(1-p_3z^2) + p_2q_1q_2q_3z^3(1-u)(1-p_1q_2z) \right)^2 \\&& - 4 p_2 q_1 q_2 q_3  z^3  (1-p_1q_2 z) \left( (1-p_1q_2z) (1-p_3z^2)  + p_2q_1 q_2q_3 z^3 (1-u)  \right).
\end{eqnarray*}
To prove the claim for $\mathbb{X}_n^{lb}$, we apply Theorem~\ref{fl}  to $G(z,u) = z^2 S(z,u)$. The functions in condition $(i)$ are $A(z,u) = P(z,u)z^2$, $B(z,u) = -\frac{1}{2p_2 q_3 (1-p_1q_2z)}$, and $C^{lb}(z,u)$. The conditions of Theorem~\ref{fl} can be checked as before by using the fact that $$C^{lb}(z,1) = R(z) (1-p_1q_2z)^2 (1-p_3z^2).$$

The proof for right bulges is exactly the same as the one for left bulges and $C^{rb}(z,u)=C^{lb}(z,u)$.

For interior loops, we set $x=y=u=v=1$ and  we get 
\[S(z,w) = Q(z,w) - \frac{\sqrt{C^{i}(z,w)}}{2p_2q_3z^2(1-p_1q_2z)}, \] where
\begin{eqnarray*}
Q(z,w)&=&\frac{(1-p_1q_2z)^2(1-p_3z^2) + p_1p_2q_1q_2^2q_3z^4(1-w)}{2p_2q_3z^2(1-p_1q_2z)},\\
C^{i}(z,w)&=&\left( (1-p_1q_2z)^3(1-p_3z^2) + p_1p_2q_1q_2^2q_3z^4(1-w)(1-p_1q_2z) \right)^2 \\&&-4 p_2 q_1 q_2 q_3  z^3  (1-p_1q_2z)^4 (1-p_3z^2) \\ &&- 4p_1p_2^2q_1^2 q_2^3 q_3^2z^7 (1-w)(1-p_1q_2 z)^2.
\end{eqnarray*}
As in the previous cases one can show that $G(z,w) = z^2 S(z,w)$ satisfies the conditions of Theorem~\ref{fl} by setting the functions in condition $(i)$ to be $A(z,w) = P(z,w)z^2$, $B(z,w) = -\frac{1}{2p_2 q_3 (1-p_1q_2z)}$, and $C^{i}(z,w)$.  Additionally, the factorization $$C^{i}(z,1) = R(z)(1-p_1q_2z)^4 (1-p_3z^2) $$ is used.

The case for multi-branch loops is similar. For completeness, we give the formula for $S(z,y)$. 
\[S(z,y) = Q(z,y) - \frac{\sqrt{C^{m}(z,y)}}{2p_2q_3z^2y(1-p_1q_2z)^2}, \] where
\begin{eqnarray*}
Q(z,y)&=&\frac{(1-p_1q_2z)^2(1-p_3z^2) - 2p_2q_1q_2q_3z^3(1-y)(1-p_1q_2z)}{2p_2q_3z^2y},\\
C^{m}(z,y)&=&\left( (1-p_1q_2z)^3(1-p_3z^2) - 2p_2q_1q_2q_3z^3(1-y)(1-p_1q_2z) \right)^2 \\&& - 4 p_2 q_1 q_2 q_3  z^3 y  (1-p_1q_2z)^4(1-p_3z^2) \\ && + 4p_2^2q_1^2 q_2^2 q_3^2 z^6 y (1-y)(1-p_1q_2 z)^2  .
\end{eqnarray*}
The claim for $\mathbb{X}_n^m$ follows from Theorem~\ref{fl} for the function $G(z,y) = z^2 S(z,y)$. Then the functions in condition $(i)$ are $A(z,y) = P(z,y)z^2$, $B(z,y) = -\frac{1}{2p_2 q_3 y (1-p_1q_2z)^2}$, and $C^{m}(z,y)$. When checking the conditions, one uses that $$C^{m}(z,1) = R(z)(1-p_1q_2z)^4 (1-p_3z^2).$$
\end{proof}

\subsection{Multi-branch loops with fixed degree}\label{fixedbranching}
In this subsection we  compute the expected number of multi-branch loops of a fixed degree $r \geq 2$. A multi-branch loop has a degree $r$ if it contains $r+1$ base pairs. 

Let $r \geq 2$ be fixed. Starting with $LS$, to get a multi-branch loop of degree $r$ with single-stranded segments of lengths $k_0, k_1, \dots, k_r$  $(k_i \geq 0)$, one needs to apply the rule $S \rightarrow LS$ exactly $r-2 + \sum_{i=0}^{r} {k_i} $ times and the rule $S \rightarrow L$ exactly once. After this one has $r+ \sum_{i=0}^{r} {k_i} $ copies of L. Then one applies the rule $L \rightarrow s$ exactly $\sum_{i=0}^{r} {k_i}$ times to get the single-stranded nucleotides, and the rule $ L \rightarrow dFd'$ precisely $r$ times to get the $r$ helices. 
Therefore, if  $z$ marks the number of nucleotides and $t$ marks the number of multi-branch loops of degree $r$,  the total weight of all substructures with $r$ branches and prescribed lengths of single-stranded segments that can be derived with this process is
\[ p_1^{r-2+\sum_{i=0}^{r} {k_i} } p_2^ {r} q_1 q_2^{\sum_{i=0}^{r} {k_i}} t z^{2r+\sum_{i=0}^{r} {k_i}} F^r\] and the total weight of all substructures starting with a multi-branch loops of degree $r$ is 
\[ \sum_{k_0, k_1, \dots, k_r \geq 0} p_1^{r-2+\sum_{i=0}^{r} {k_i} } p_2^ {r} q_1 q_2^{\sum_{i=0}^{r} {k_i}} t z^{2r+\sum_{i=0}^{r} {k_i}} F^r = \frac{ p_1^{r-2} p_2^r q_1 t z^{2r} F^{r} } { (1-p_1 q_2 z)^{r+1}}.\]
Translation of the grammar into generating functions yields the system
\begin{eqnarray}
 S &=& p_1 LS + q_1 L, \label{m1}\\
 L &=& p_2 z^2 F + q_2 z, \label{m2}\\
F &=& p_3 z^2 F + q_3     \frac{ p_1^{r-2} p_2^r q_1 z^{2r} t F^{r} } { (1-p_1 q_2 z)^{r+1} }+ q_3 \left( LS - \frac{p_1^{r-2} p_2^r q_1  z^{2r} F^{r} } { (1-p_1 q_2 z )^{r+1}}\right). \label{m3} \end{eqnarray}
 For convenience, set \[T_r(z) =  \frac{p_1^{r-2} p_2^r q_1 q_3  z^{2r}  } { (1-p_1 q_2 z )^{r+1}}. \]
Then equation~\eqref{m3} can be rewritten as
\begin{equation} \label{m4} F =  p_3 z^2 F  + (t -1)T_r F^r + q_3 LS .\end{equation}
Multiplying  equations~\eqref{m1} and~\eqref{m2} we get
\[ LS = p_1LS (p_2 z^2 F + q_2 z) + q_1 (p_2 z^2 F + q_2 z)^2 \]
and hence
\[LS = \frac{q_1(p_2 z^2 F + q_2 z)^2}{1-p_1p_2 z^2 F - p_1q_2 z}. \]
Substituting back to~\eqref{m4}, we get:
\[F =  p_3 z^2 F  + (t -1)T_r F^r + q_3  \frac{q_1(p_2 z^2 F + q_2 z)^2}{1-p_1p_2 z^2 F - p_1q_2 z}. \]
which is equivalent to
\begin{eqnarray*} p_1p_2z^2(t-1) T_r F^{r+1} -(t-1)(1 - p_1q_2 z)T_rF^r + ( p_1p_2p_3 z^4 - p_2^2 q_1 q_3 z^4 -p_1p_2 z^2) F^2  \\
\hspace{0.2cm} + (1-p_1q_2 z- p_3 z^2 + p_1  p_3q_2 z^3  - 2p_2q_1q_2q_3 z^3) F - q_1q_2^2 q_3z^2=0. \end{eqnarray*} 
After differentiating with respect to $t$, we find that $F_t'(z,1)$ is equal to
\begin{equation}\label{f'}  \frac{T_r F^r(z,1) (1  - p_1q_2 z -p_1p_2 z^2 F(z,1))}{2 p_2z^2[( p_1p_3-p_2 q_1 q_3) z^2-p_1]F(z,1) + [ q_2(p_1  p_3- 2p_2q_1q_3) z^3- p_3 z^2 -p_1q_2 z+1]} \end{equation}
and from here we can easily find  the function $F$ at $t=1$, which we will need later. Namely, 
\begin{eqnarray*} && p_2z^2( p_1  - p_1 p_3 z^2 + p_2 q_1 q_3 z^2) F^2(z,1) \\
&& -(1- p_1 q_2 z - p_3 z^2 + p_1 p_3 q_2 z^3 - 2 p_2 q_1 q_2 q_3 z^3 )F(z,1)  + q_1 q_2 ^2 q_3 z^2 =0 \end{eqnarray*}
and hence after simplifications we find that
\begin{equation} \label{fz} F(z) = \frac{(1- p_1 q_2 z - p_3 z^2 + p_1 p_3 q_2 z^3 - 2 p_2 q_1 q_2 q_3 z^3) - \sqrt{(1-p_3z^2)R(z)}}{2 p_2 z^2  ( p_1  - p_1  p_3 z^2 + p_2 q_1 q_3 z^2)}. \end{equation}
The solution with the negative sign is chosen because $F(0)=0$. From the explicit formula for $F(z)$, we note that the dominant singularity is again $\rho_0$. Indeed, $F(z)$ has a positive dominant singularity, since it has positive coefficients and if $z_0 < \rho_0$ is a positive solution to the quadratic  $p_1  - p_1  p_3 z^2 + p_2 q_1 q_3 z^2 =0$, we get
\begin{eqnarray*}&&(1- p_1 q_2 z+0 - p_3 z_0^2 + p_1 p_3 q_2 z_0^3 - 2 p_2 q_1 q_2 q_3 z_0^3) - \sqrt{(1-p_3z_0^2)R(z)} \\&&{} = (1-p_3z_0^2)(1+p_1q_2z_0)-\sqrt{(1-p_3z_0^2)^2(1+p_1q_2z_0)^2} =0. \\ \end{eqnarray*}
Combining~\eqref{m1} and~\eqref{m4} yields
\begin{equation*}S= \frac{p_1}{q_3} ( F -  p_3 z^2 F - (t -1)T_r F^r) + q_1 (p_2 z^2 F + q_2 z) \end{equation*}
and hence
\begin{equation} \label{s'} S'_t(z,1) = \frac{p_1}{q_3} ( F'_t(z,1) -  p_3 z^2 F'_t(z,1) - T_r F^r(z,1)) +  p_2 q_1 z^2 F'_t(z,1). \end{equation}
In light of~\eqref{fz}, formula~\eqref{f'} simplifies to
\[F_t'(z,1) = \frac{T_r F^r(z,1) (1  - p_1q_2 z -p_1p_2 z^2 F(z,1))}{\sqrt{(1-p_3z^2)R(z)}}\]
and plugging this into~\eqref{s'} yields
\[S_t'(z,1) = \frac{T_r F^r(z,1)}{2q_3}\left( \frac{p_1^2p_3q_2z^3-p_1p_3z^2+2p_2q_1q_3z^2-p_1^2q_2z+p_1}{\sqrt{(1-p_3z^2)R(z)}} -p_1\right).\]
Using this expression and Theorem~\ref{fo}, we can estimate the coefficients of $S_t'(z,1)$:
\begin{equation}\label{sn'} [z^n]S'(z,1) \sim \frac{K}{ \Gamma(1/2)} n^{1/2} \rho_0^{-n}, \end{equation}
where \[ K=  \frac{p_1^{r-2}q_1q_2^{r-1}\rho_0^{r-1}(1-p_3\rho_0^2)}{ 4\sqrt{-\rho_0 (1-p_3 \rho_0^2) R'(\rho_0)}(1+p_1q_2\rho_0)^{r-1}}.\]
 Combining this estimate with~\eqref{sn} we get the following theorem.
 \begin{theorem} Let $\mathbb{X}_n^{m,r}$ be the number of multi-branch loops of degree $r$ in a random secondary structure with $n$ nucleotides. If the probabilities $p_i, q_i$ are all non-zero, then
  \[ \mathbb{E}(\mathbb{X}_n^{m,r}) \sim \frac{p_1^{r-2}q_2^{r-2}\rho_0^{r-2}(1-p_1q_2\rho_0)(1-p_3\rho_0^2)}{4(1+p_1q_2\rho_0)^{r-1}(3-p_1q_2\rho_0-p_3\rho_0^2-p_1p_3q_2\rho_0^3)}  n.\]
 \end{theorem}
 \begin{proof} The estimate follows from~\eqref{sn'},~\eqref{sn}, and $\mathbb{E}(\mathbb{X}_n^{m,r}) =\frac{[z^n]S'_t(z,1)}{[z^n]S(z,1)}.$
 \end{proof}

\subsection{External loop}~\label{externalloop} In this subsection we analyze the branchings of the external loop and the 5'-3' distance. The 5'-3' distance is defined as the number of nucleotides (paired or single-stranded) enclosed in the external loop minus one. Let $u$ be the variable that marks the number of helices in the external loop, and let $v$ mark the 5'-3' distance.
The total contribution  of all  secondary structures with no base pairs in $S(z,u,v)$ is
\[\sum_{n \geq 1} P(S  \ras s^n) = \sum_{n \geq 1} p_1^{n-1} q_1 q_2^n z^n v^{n-1} = \frac{q_1 q_2z}{1-p_1q_2zv}.\]
All other structures have $r \geq 1$ helices in the external loop. Since
\[P( S \ras s^{k_0} dFd' s^{k_1} \cdots dFd' s^{k_r}) = p_1^{r-1+\sum_{i=0}^{r} {k_i} } p_2^ {r} q_1 q_2^{\sum_{i=0}^{r} {k_i}},\]
the generating function of all structures that have exactly $r$ helices in the external loop is given by
\begin{equation*} \sum_{k_0, k_1, \dots, k_r \geq 0} p_1^{r-1+\sum_{i=0}^{r} {k_i} } p_2^ {r} q_1 q_2^{\sum_{i=0}^{r} {k_i}} z^{2r+\sum_{i=0}^{r} {k_i}}u^r v^{2r-1+\sum_{i=0}^{r} {k_i}}F^r(z)  \end{equation*}
which is equal to $\frac{p_1^{r-1}p_2^r q_1z^{2r} u^r v^{2r-1}F^r(z)}{(1-p_1q_2zv)^{r+1}}.$
Therefore $S(z,u,v)$ is given by
\begin{eqnarray*}S(z,u,v) &=& \frac{q_1 q_2z}{1-p_1q_2zv} + \sum_{r \geq 1}\frac{p_1^{r-1}p_2^r q_1z^{2r} u^r v^{2r-1} F^r(z)}{(1-p_1q_2z v)^{r+1}} \\
&=&  \frac{q_1 q_2z}{1-p_1q_2zv} + \frac{p_2q_1z^2uvF(z)}{(1-p_1q_2zv)(1-p_1q_2zv-p_1p_2z^2uv^2F(z))}.\end{eqnarray*}
To compute the expected number of helices in the external loops we will need to look at the behavior of $S'_u(z,1,1)$ around its dominant singularity. We find that
\[S'_u(z,1,1) = \frac{p_2 q_1 z^2  F(z)}{(1-p_1q_2z-p_1p_2z^2F(z))^2}. \]
Using~\eqref{fz}, one can show that $1-p_1q_2z-p_1p_2z^2F(z) \neq 0$, and so the dominant singularity of $S'_u(z,1,1)$ is the same as the dominant singularity of $F(z)$, which was found to be $\rho_0$. 
After simplifications of the expansion of  $S'_u(z,1,1)$, we get that as $z \rightarrow \rho_0$,
\begin{equation} \label{su'} S'_u(z,1,1) \sim -\frac{(1+2p_1q_2\rho_0)\sqrt{-\rho_0 R'(\rho_0)(1-p_3\rho_0^2)}} {2p_1q_3\rho_0^2}\left(1-\frac{z}{\rho_0}\right)^{1/2}.\end{equation}

 \begin{theorem} Let $\mathbb{X}_n^{eh}$ be a random variable counting the number of helices in the external loop in a secondary structure with $n$ nucleotides and let $\mathbb{X}_n^{ecd}$ count the 5'-3' distance. If the probabilities $p_i, q_i$ are all non-zero, then
  \begin{eqnarray*} \mathbb{E}(\mathbb{X}_n^{eh}) &\sim& 1+2p_1q_2\rho_0 \hspace{2cm} \text{and} \\ 
  \mathbb{E}(\mathbb{X}_n^{ecd}) &\sim& \frac{1+5p_1q_2\rho_0-2p_1^2q_2^2\rho_0^2}{1-p_1q_2\rho_0}.
  \end{eqnarray*}
  
 \end{theorem}
 \begin{proof}  The estimate for $\mathbb{E}(\mathbb{X}_n^{eh})$ follows from~\eqref{su'},~\eqref{sn}, Theorem~\ref{fo}, and the fact that $\mathbb{E}(\mathbb{X}_n^{eh}) =\frac{[z^n]S'_u(z,1,1)}{[z^n]S(z,1,1)}.$ For $ \mathbb{E}(\mathbb{X}_n^{ecd})$, one finds that 
  \begin{eqnarray*}S'_v(z,1,1) = \frac{p_1q_1q_2^2z^2}{1-p_1q_2z} &+& \frac{p_2 q_1 z^2 (1-p_1q_2z)(1-p_1q_2z +2p_1q_2z) F(z)}{(1-p_1q_2z)^2(1-p_1q_2z-p_1p_2z^2F(z))^2} \\&+& \frac{p_1p_2^2q_1z^4(1-2p_1q_2z)F(z)^2}{(1-p_1q_2z)^2(1-p_1q_2z-p_1p_2z^2F(z))^2}.\end{eqnarray*}
 The dominant singularity is again $\rho_0$, so one proceeds as before to obtain the estimate.
 \end{proof}
 
\subsection{The function $f$ in Theorem~\ref{mainth}} \label{funf}
In this subsection we show that the set of probabilities $(p_1, p_2, p_3)$ for which Theorem~\ref{mainth} does not apply is small in the sense that it has Lebesgue measure zero. Define $V^{bp} = V^{bp}(p_1, p_2,p_3, \rho_0)$ to be
 $$\rho_0 (C^{bp}_{1,0})^2 C^{bp}_{0,2} - 2 \rho_0 C^{bp}_{1,0} C^{bp}_{1,1} C^{bp}_{0,1} + \rho_0 C^{bp}_{2,0}(C^{bp}_{0,1})^2  + (C^{bp}_{0,1})^2 C^{bp}_{1,0} + \rho_0 C^{bp}_{0,1} (C^{bp}_{1,0})^2 \bigg|_{\substack{z=\rho_0\\ u=1}} $$
 where $C^{bp}$ is the polynomial that is defined in~\eqref{Cbp} and appears in the conditions of Theorems~\ref{bps}. Similarly define $V^{hel}, V^{hp}, V^{lb}, V^{i}, V^{m}$ which correspond to the polynomials $C^{hel}, C^{hp}, C^{lb}, C^{i}$, and $C^{m}$, which appear in the conditions of Theorems~\ref{hel} and~\ref{lps} (since $C^{lb}=C^{rb}$, we do not need to define $V^{rb}$). Finally, define $$g(p_1, p_2, p_3, \rho_0) =V^{bp} V^{h} V^{hp} V^{lb} V^{i} V^{m}.$$ Notice that Theorem~\ref{mainth} holds for all $(p_1, p_2, p_3) \in (0,1)^3$ other than those for which $g(p_1, p_2, p_3, \rho_0) =0$. Since by Lemma~\ref{lm} $\rho_0$ is a root of multiplicity one of the polynomial $R(z)$ for all $(p_1, p_2, p_3) \in (0,1)^3$ it follows that $\rho_0$ is an analytic function of $(p_1, p_2, p_3)$ and therefore $$f(p_1, p_2, p_3) = g(p_1, p_2, p_3, \rho_0 (p_1, p_2, p_3))$$ is also analytic on $(0,1)^3$. This implies that its zero set must be of measure zero and hence the central limit results hold for almost all choices of the grammar probabilities.

\section{Discussion} \label{discussion}
Recall that $\mathbb{X}_n^{lb}, \mathbb{X}_n^{rb}, \mathbb{X}_n^{m}, \mathbb{X}_n^{hel}, \mathbb{X}_n^{hp}$, and $\mathbb{X}_n^{i}$ are the number of  left bulges, right bulges, multi-branch loops, helices, hairpins, and internal loops in a random secondary structure on $n$ nucleotides, respectively, while $\mathbb{X}_n^{m, r}$ is the number of multi-branch loops of degree $r$. Since $p_1q_2\rho_0 < 1$, based on the calculated  expectations, we have the following corollary.
\begin{corollary}\label{col}If $p_i, q_i >0$
\begin{itemize}
\item [(i)] $\mathbb{E}(\mathbb{X}_n^{lb}) = \mathbb{E}(\mathbb{X}_n^{rb})$,
\item [(ii)] $\mathbb{E}(\mathbb{X}_n^m) = \frac{1}{4}\mathbb{E}(\mathbb{X}_n^{hel})(1+o(n))$,
\item [(iii)] $\mathbb{E}(\mathbb{X}_n^{hp})=(\mathbb{E}(\mathbb{X}_n^{i}) + \mathbb{E}(\mathbb{X}_n^{m})) (1+ o(n))$,
\item [(iv)] $\mathbb{E}(\mathbb{X}_n^{m})=(\mathbb{E}(\mathbb{X}_n^{lb}) + \mathbb{E}(\mathbb{X}_n^{i})) (1+ o(n))$,
\item [(v)] $\mathbb{E}(\mathbb{X}_n^{m, r+1}) < \frac{1}{2}\mathbb{E}(\mathbb{X}_n^{m, r})(1+o(n))$, $r \geq 2$.
\end{itemize}
\end{corollary}

Note that these relations hold even for the probabilities for which the function $f$ discussed in Section~\ref{funf} is zero. Namely, the means in those cases can be computed using Theorem~\ref{fo} and calculations similar to the ones in Sections~\ref{fixedbranching} and~\ref{externalloop}. The asymptotic formulas for the expected number of base pairs, helices, and loops remain the same as for the generic probabilities.

%Consequently, even though the means and the standard deviations of the random variables analyzed depend on the grammar probabilities, the ratios of the different types of motifs do not. 

%Consider now the probability space of all structures $(str_i, seq_j)$ with $n$ nucleotides, where $seq_j$ is a sequence of $n$ nucleotides and $str_i$ is the set of base pairs, generated by the Knudsen-Hein grammar. So far, we were considering only the base pair structure without the underlying nucleotide sequence. However, since the motifs we considered are determined only by the base pairs, the distributions we have described and the expected values we have computed are the same even in this larger sample space. For example, for the expected number of helices, we have

%\[\mathbb{E}(\mathbb{X}_n^{h}) = \sum_{i} \mathbb{X}_n^{h} (str_i) P(str_i) = \sum_{i,j} \mathbb{X}_n^{h} (str_i) P(str_i, seq_j), \]
%where the first sum is taken over all structures of length $n$ and the second sum is taken over all combinations of structures and sequences of length $n$.

When an SCFG for RNA secondary structure prediction is constructed, the goal is to adequately describe the objects of interest, in this case the native RNA secondary structures. The default parameters used in Pfold were obtained by an expectation maximization procedure on a training set of tRNA and large subunit ribosomal RNA secondary structures~\citep{SCFG}. The transmission probabilities  are $$p_1=0.868534,  p_2=0.105397, p_3=0.787640,$$ $$q_1=0.131466, q_2=0.894603,  q_3=0.212360$$ and the emission probabilities are given in Table~\ref{emission}. 
\begin{table}[hbt!]
\begin{tabular}{c| c c c c}
& A & U & G &C \\
\hline
A & 0.001167 & 0.177977 & 0.001058 & 0.001806 \\
U & 0.177977 & 0.002793 & 0.049043 & 0.000763 \\
G & 0.001058 & 0.049043 & 0.000406 & 0.266974 \\
C & 0.001806 & 0.000763 & 0.266974 & 0.000391
\end{tabular}
\hspace{1.5cm}
\begin{tabular}{c|c}
A & 0.364097\\
U & 0.273013\\
G & 0.211881\\
C & 0.151009
\end{tabular}
\caption{The default paired and unpaired probabilities used in Pfold}
\label{emission}
\end{table}

These probabilities generate a distribution that should describe the training set as a whole. So, a natural question is how well it describes particular classes of RNA that may or may not have been used in the training.  Ideally, the native structures should be the likely ones among all the possible secondary structures and the average number of motifs observed should be close to the expectation given by the model. 

To compare known RNA structures to the asymptotic expected distributions from the model, we downloaded all 854 5S, 16S, and 23S ribosomal structures from the Comparative RNA website~\citep{CRW} for which the secondary structure (without pseudoknots) has been determined by covarying sequence analysis and given in a .ct file. Out of those we selected the structures which do not have ambiguous nucleotides and this left us with a final  set of 400 structures. From these we selected 5 sets of sequences, for each of which the variance in length is small. Each of the five sets consists of sequences of the same type with approximately the same secondary structure. Their composition is given in Table~\ref{fivesets}. The average numbers of various motifs and their standard deviations for the comparative structures  of the sequences in each set are given in Table~\ref{avstdev} and Table~\ref{loopsavstdev} in the rows labeled CRW.

\begin{table} [t!]
\begin{center}
\begin{tabular} {|l|c|c|c|} 
\hline
& No. Sequences (Type) &	Av. Length &	St. Dev. Length 	\\
\hline
\hline
Set I & 122 (5S) &	121.17 & 3.1 	\\
\hline
%Set II & 38 (37 16S, 1 23S) & 956.37 & 6.45 \\	
Set II & 37 (16S) & 956.46 & 6.51 \\	
\hline
Set III & 81 (16S)  & 1521.33 & 24.86 \\
\hline
Set IV & 50 (16S)& 1787.1 & 20.9 \\
\hline
Set V & 34 (23S) & 2912.85 & 23.08\\
\hline
\end{tabular}
\caption{Five sets of RNA sequences chosen from the CRW database to minimize variance in sequence length}
 \label{fivesets}
 \end{center}
\end{table}

\begin{table}[b!]
\begin{center}
\begin{tabular}{|c|c|c|c|}
\cline{3-4}
\multicolumn{2}{c}{} & \multicolumn{2}{|c|}{Averages and Standard Deviations} \\
\cline{3-4}
\multicolumn{2}{c}{} &   \multicolumn{1}{|c|}{BP}  & Hel  \\
\cline{3-4} \hline
\multirow{3}{*}{\begin{sideways}Set I \end{sideways}}& CRW  & 38.00 $\pm$ 2.13 & 7.96 $\pm$ 0.39  \\
    \cline{2-4}
                         &CYK & 34.05 $\pm$ 4.13 & 5.64 $\pm$ 0.76  \\
                          \cline{2-4}
                         &Model & 32.83 $\pm$ 16.90 & 6.96 $\pm$ 1.53 \\
    \hline  \hline
\multirow{3}{*}{\begin{sideways}Set II \end{sideways}}& CRW & 254.65 $\pm$ 2.38 & 65.84 $\pm$ 0.65  \\
    \cline{2-4}
                         &CYK & 206.08 $\pm$ 10.31 & 33.65 $\pm$ 2.51 \\
                          \cline{2-4}
                         &Model & 259.12 $\pm$ 47.47 & 54.96 $\pm$ 4.29 \\
    \hline \hline
\multirow{3}{*}{\begin{sideways}Set III \end{sideways}}& CRW & 457.54 $\pm$ 9.68 & 104.16 $\pm$ 3.08 \\
    \cline{2-4}
                         &CYK & 423.69 $\pm$ 23.58 & 68.76 $\pm$ 4.47 \\
                          \cline{2-4}
                         &Model & 412.16 $\pm$ 59.87 & 87.41 $\pm$ 5.41 \\
    \hline \hline
\multirow{3}{*}{\begin{sideways}Set IV \end{sideways}}& CRW & 485.84 $\pm$ 10.50 & 112.9 $\pm$ 3.62  \\
    \cline{2-4}
                         &CYK & 493.12 $\pm$ 16.73 & 76.00 $\pm$ 3.17 \\
                          \cline{2-4}
                         &Model & 484.16 $\pm$ 64.89 & 102.68 $\pm$ 5.86 \\
    \hline \hline
\multirow{3}{*}{\begin{sideways}Set V \end{sideways}}& CRW & 837.94 $\pm$ 8.99 & 199.18 $\pm$ 2.91 \\
    \cline{2-4}
                         &CYK &  795.68 $\pm$ 29.01 &  131.91 $\pm$ 7.23 \\
                          \cline{2-4}
                         &Model & 787.79 $\pm$ 82.84& 167.37 $\pm$ 7.48 \\
    \hline
     \end{tabular}
   \caption{Average number of base pairs and helices in structures for the sequences in the five sets. For each set,  the structures from the CRW database and the structures predicted using the CYK algorithm are considered. In addition, the expected number of motifs for sequences of the same length is given for the model using the default Pfold probabilities, where $n$ is taken to be the average length of the sequences in the corresponding set}
   \label{avstdev}
   \end{center}
\end{table}

Our goal now is to see how well the Knuden-Hein grammar can describe the ribosomal structures. To that end we folded each of these sequences using our implementation of the Cocke-Younger-Kasami (CYK) algorithm and the default Pfold probabilities for all the sequences in our five sets. The results of the CYK parsing are also displayed in Table~\ref{avstdev} and Table~\ref{loopsavstdev}. 

We remark that the sequences were folded using the CYK algorithm, which computes the most probable structure for a given sequence. In contrast, the Pfold program predicts the structure with the highest expected number of correctly predicted positions. As a consequence, such predicted structures have very few base pairs. For example, the structures obtained by using PPfold~\citep{PPfold} (a parallelized version of Pfold) for the sequences in Set~I have on average 15.78 base pairs. This is because Pfold has been designed to 
be used primarily for finding a consensus structure for a set of aligned sequences and this is where its strength lies.

In addition, Tables~\ref{avstdev} and~\ref{loopsavstdev} include the expectations that correspond to the average sequence lengths of  our  five sets and the default Pfold parameters. It is not surprising that  the CYK predictions do not agree with the model since our results describe the distributions of the motifs in structures over random sequences, while the biological sequences are not random. Nonetheless, we observe that the average number of base pairs in the CRW structures is within one standard deviation of the model mean. This is also true for most of the CYK structures. So, we conclude that the grammar describes well the number of base pairs in the ribosomal structures.
\begin{table}[htb!]
\begin{center}
\begin{tabular}{|c|c|c|c|c|c|c|}
\cline{3-7}
\multicolumn{2}{c}{} & \multicolumn{5}{|c|}{Averages and Standard Deviations} \\
\cline{3-7}
\multicolumn{2}{c}{} &   \multicolumn{1}{|c|} {ML} & IL & LB & RB &HL \\
\cline{3-7} \hline
\multirow{3}{*}{\begin{sideways}Set I \end{sideways}}& CRW  & 1.00 $\pm$ 0.00 & 2.02 $\pm$ 0.13 & 1.99 $\pm$ 0.20 & 0.95 $\pm$ 0.31 & 2.00 $\pm$ 0.00 \\
    \cline{2-7}
                         &CYK & 1.07 $\pm$ 0.47 & 1.54 $\pm$ 0.81  & 0.39 $\pm$ 0.55 & 0.30 $\pm$ 0.49 & 2.34 $\pm$ 0.71 \\
                          \cline{2-7}
                         &Model  & 1.74 $\pm$ 0.76 & 1.35 $\pm$ 1.29 & 0.39 $\pm$ 0.65 & 0.39 $\pm$ 0.65 & 3.09 $\pm$ 1.33\\
    \hline  \hline
\multirow{3}{*}{\begin{sideways}Set II \end{sideways}}& CRW  & 10.00 $\pm$ 0.00 & 18.57 $\pm$ 0.50 & 7.24 $\pm$ 0.43 & 9.03 $\pm$ 0.16 & 21.00 $\pm$ 0.00 \\
    \cline{2-7}
                         &CYK & 8.84 $\pm$ 1.54 & 9.22 $\pm$ 2.50 & 1.76 $\pm$ 1.01 & 0.89 $\pm$ 0.94 & 15.95 $\pm$ 2.03\\
                          \cline{2-7}
                         &Model & 13.74 $\pm$ 2.14 & 10.68 $\pm$ 3.63 & 3.06 $\pm$ 1.83 & 3.06 $\pm$ 1.83 & 24.42 $\pm$ 3.73\\
    \hline \hline
\multirow{3}{*}{\begin{sideways}Set III \end{sideways}}& CRW & 16.86 $\pm$ 0.38 & 27.86 $\pm$ 1.12 & 9.49 $\pm$ 1.06 & 18.07 $\pm$ 1.29 & 31.86 $\pm$ 0.38\\
    \cline{2-7}
                         &CYK  & 17.65 $\pm$ 2.03 & 15.99 $\pm$ 4.03 & 2.98 $\pm$ 1.70 & 2.86 $\pm$ 1.79 & 29.27 $\pm$ 3.04 \\
                          \cline{2-7}
                         &Model  & 21.85 $\pm$ 2.70 & 16.98 $\pm$ 4.56 & 4.87 $\pm$ 2.31 & 4.87 $\pm$ 2.31 & 38.84 $\pm$ 4.71\\
    \hline \hline
\multirow{3}{*}{\begin{sideways}Set IV \end{sideways}}& CRW  & 16.00 $\pm$ 0.20 & 33.68 $\pm$ 3.20 & 11.86 $\pm$ 1.14 & 18.38 $\pm$ 1.26 & 32.98 $\pm$ 0.32 \\
    \cline{2-7}
                         &CYK  & 19.06 $\pm$ 2.12 & 16.02 $\pm$ 3.07 & 4.38 $\pm$ 1.93& 4.22 $\pm$ 1.66 & 32.32 $\pm$ 2.77\\
                          \cline{2-7}
                         &Model  & 25.67 $\pm$ 2.93 & 19.95 $\pm$ 4.97 & 5.72 $\pm$ 2.50 & 5.72 $\pm$ 2.50 & 45.62 $\pm$ 5.10\\
    \hline \hline
\multirow{3}{*}{\begin{sideways}Set V \end{sideways}}& CRW  & 33.65 $\pm$ 1.30 & 53.15 $\pm$ 2.02 & 23.5 $\pm$ 1.33 & 18.41 $\pm$ 1.81 & 70.47 $\pm$ 0.71\\
    \cline{2-7}
                         &CYK &  35.97 $\pm$ 3.11&  25.41 $\pm$ 5.08&  6.59 $\pm$ 2.18&  6.74 $\pm$ 2.34&  57.21 $\pm$ 4.14\\
                          \cline{2-7}
                         &Model  & 41.84 $\pm$ 3.74 & 32.52 $\pm$ 6.34 & 9.32 $\pm$ 3.20 & 9.32 $\pm$ 3.20& 74.36 $\pm$ 6.51\\
    \hline
\end{tabular}
   \caption{Average number of loops in structures for the CRW and CYK structures for the sequences in the five sets. The model averages are computed using the default Pfold probabilities and  $n$ is taken to be the average length of the sequences in the corresponding set}
   \label{loopsavstdev}
   \end{center}
\end{table}

However, the way these base pairs are arranged  in the CRW  and CYK structures is noticeably different. The CYK structures have fewer helices that are longer on average as can be seen from Table~\ref{avstdev}. On the other hand, the CRW structures have shorter helices separated by internal loops and bulges, which form stable stems, while branching is less favored. Namely, it can be seen from the ratios in the last two columns of Table~\ref{ratios} that internal loops and bulges occur more frequently relatively to the multi-branch loops in the CRW structures than in the CYK structures. 

\begin{table}[hbt!]
\begin{center}
{
\renewcommand{\arraystretch}{1.3}
\begin{tabular}{|c|c|c|c|c|c|c|}
\cline{3-7}
\multicolumn{2}{c}{} & \multicolumn{5}{|c|}{Ratios of Averages} \\
\cline{3-7}
\multicolumn{2}{c}{} &   \multicolumn{1}{|c|} {$\frac{\text{RB}}{\text{LB}}$} &  $\frac{\text{Hel}}{\text{ML}}$& $\frac{\text{IL} + \text{ML}}{\text{HP}}$ & $\frac{\text{IL} + \text{LB}}{\text{ML}}$ &$\frac{\text{IL} + \text{RB}}{\text{ML}}$ \\
\cline{3-7} \hline
\multirow{2}{*}{\begin{sideways}Set I \end{sideways}}& CRW  & 0.48 & 7.96 & 1.51 & 4.01 & 2.97\\
    \cline{2-7}
                         &CYK & 0.77 & 5.27 & 1.12 & 1.80 & 1.72 \\
    \hline  \hline
\multirow{2}{*}{\begin{sideways}Set II \end{sideways}}& CRW  & 1.25 & 6.58 & 1.36 & 2.58 &2.76\\
    \cline{2-7}
                         &CYK & 0.50 &  3.80 & 1.13 & 1.24 & 1.14\\    \hline \hline
\multirow{2}{*}{\begin{sideways}Set III \end{sideways}}& CRW & 2.50 & 6.18 & 1.40 & 2.22 & 2.72 \\
    \cline{2-7}
                         &CYK  & 0.96 & 3.90 & 1.15 & 1.07 & 1.07 \\
                             \hline \hline
\multirow{2}{*}{\begin{sideways}Set IV \end{sideways}}& CRW  &  1.93 & 7.06 &1.51 & 2.85 & 3.25 \\
    \cline{2-7}
                         &CYK  & 0.96 & 3.99 & 1.09 & 1.07 & 1.06\\
                             \hline \hline
\multirow{2}{*}{\begin{sideways}Set V \end{sideways}}& CRW  & 0.78 & 5.92 & 1.23 & 2.28 & 2.13\\
    \cline{2-7}
                         &CYK &  1.02 & 3.67 & 1.07 & 0.89 & 0.89\\
                            \hline \hline
                            
                   \multicolumn{2}{|c|}{Model} & 1 & 4& 1& 1& 1\\
                            \hline
    \end{tabular}
   \caption{Ratios of the average number of occurrences of various motifs for the native and predicted structures from the five sets. The last row contains the asymptotic model averages as given by Corollary~\ref{col} }
   \label{ratios}
 }
   \end{center}
\end{table}

The ratios given in Table~\ref{ratios} clearly indicate the differences between the CRW and the CYK structures, as well as the agreement of the CYK structures with the model averages.  For example, CYK predicts approximately the same number of left and right bulges, while they occur with different frequencies in the CRW structures. The  agreement between the ratios for the CYK and the model is especially striking for the Sets~III-V, which contain longer sequences, and is more expected because our results are asymptotic. This suggests that even though the grammar probabilities can be adjusted to, say,  increase the number of helices in the CYK structures, the relative frequencies of the loops  in the CYK structures will remain close to the model predictions, which are independent of the parameters. Therefore, we expect that the change of the grammar probabilities will not improve the CYK prediction of  structures for the sequences in the Sets III-V significantly. Given that the CRW structures for these sequences are long and complex, it would be interesting to see whether there are grammars which  reflect their branching behavior more closely, while still being simple enough for computational purposes.

\section*{Acknowledgements}

The authors would like to thank Christian Reidys for  useful comments on an earlier version of these results and David Esposito for implementing the CYK parsing and running the predictions.

\bibliographystyle{plainnat} 
\bibliography{scfg.bib}

% Flajolet's book
%pages = {xiv+810},
% ISBN = {978-0-521-89806-5},
% MRCLASS = {05-02 (05A15 05A16 60C05 60E10 82-01)},
% MRNUMBER = {2483235 (2010h:05005)},

\end{document}